\documentclass[letterpaper, 10 pt, conference]{ieeeconf}  
\IEEEoverridecommandlockouts
\usepackage{balance}
\usepackage{color}
\usepackage{caption}
\usepackage{subcaption}
\usepackage{enumerate}
\usepackage{cite}
\usepackage{empheq}
\usepackage{mathrsfs}

\usepackage{mathtools}

\usepackage{tikz}
\usepackage{circuitikz}


\makeatletter
\pgfcircdeclarebipole{}{\ctikzvalof{bipoles/interr/height 2}}{spst}{\ctikzvalof{bipoles/interr/height}}{\ctikzvalof{bipoles/interr/width}}{

    \pgfsetlinewidth{\pgfkeysvalueof{/tikz/circuitikz/bipoles/thickness}\pgfstartlinewidth}

    \pgfpathmoveto{\pgfpoint{\pgf@circ@res@left}{0pt}}
    \pgfpathlineto{\pgfpoint{.6\pgf@circ@res@right}{\pgf@circ@res@up}}
    \pgfusepath{draw}   
}

\def\pgf@circ@spst@path#1{\pgf@circ@bipole@path{spst}{#1}}
\tikzset{switch/.style = {\circuitikzbasekey, /tikz/to path=\pgf@circ@spst@path, l=#1}}
\tikzset{spst/.style = {switch = #1}}
\makeatother

\makeatletter
\let\proof\@undefined                        
\let\endproof\@undefined                  
\makeatother
\usepackage{graphicx,amssymb,amstext,amsmath,amsthm}

\usepackage[bookmarks=true]{hyperref}
\usepackage{algorithm,algorithmicx,algpseudocode}
\algnewcommand{\algorithmicgoto}{\textbf{go to}}%
\algnewcommand{\Goto}[1]{\algorithmicgoto~\ref{#1}}%
\algnewcommand{\LineComment}[1]{\Statex \(\triangleright\) #1}
\algnewcommand{\LineCommentN}[1]{\Statex \hspace{1cm}\(\triangleright\) #1}

\usepackage{multirow}
\usepackage{stfloats}

\newtheorem{prop}{Proposition} 
\newtheorem{cor}{Corollary}
\newtheorem{thm}{Theorem}
	\newtheorem{assumption}{Assumption}
\newtheorem{lem}{Lemma}
\newtheorem{defn}{Definition}
\newtheorem{rem}{Remark}
\newtheorem{problem}{Problem}

\usepackage{setspace}

\let\oldbibliography\thebibliography
\renewcommand{\thebibliography}[1]{%
  \oldbibliography{#1}%
}

\newcommand{\yong}[1]{{\color{black} #1}}
\newcommand{\moh}[1]{{\color{black} #1}}
\newcommand{\yongn}[1]{{\color{black} #1}}
\newcommand{\yongs}[1]{{\color{black} #1}}
\newcommand{\yongz}[1]{{\color{black} #1}}
\newcommand{\moha}[1]{{\color{black} #1}}
\newcommand{\moham}[1]{{\color{black} #1}}
\newcommand{\mohamm}[1]{{\color{black} #1}}
\newcommand{\mohk}[1]{{\color{black} #1}}
\newcommand{\syong}[1]{{\color{black} #1}}
\begin{document}

\title{\LARGE \bf Simultaneous Input and State Interval Observers\\ for Nonlinear Systems \mohamm{with \yongz{Rank-Deficient} 
Direct Feedthrough}} 

\author{%
Mohammad Khajenejad, Sze Zheng Yong\\
\thanks{
M. Khajenejad and S.Z. Yong are with the School for Engineering of Matter, Transport and Energy, Arizona State University, Tempe, AZ, USA (e-mail: \{mkhajene, szyong\}@asu.edu).}
\thanks{This work is partially supported by NSF grant CNS-1932066.}
}

\maketitle
\thispagestyle{empty}
\pagestyle{empty}

\begin{abstract}
We address the problem of designing simultaneous input and state interval observers 
for Lipschitz \yong{continuous} nonlinear systems with \yongz{rank-deficient feedthrough,} unknown inputs and bounded noise signals. Benefiting from the existence of nonlinear decomposition functions and affine abstractions, \yongs{
our proposed observer recursively computes the maximal and minimal elements of the estimate intervals that are proven to contain the true states and unknown inputs. 
Moreover, we provide \mohamm{necessary and} sufficient conditions for the 
existence and \mohamm{sufficient conditions for the} stability (i.e., uniform boundedness of the sequence of estimate interval widths) of the designed observer, 
and} show that the input interval estimates are tight, given the state intervals and decomposition functions. 
 \end{abstract}

\vspace{-0.125cm}
\section{Introduction} \vspace{-0.05cm}
\emph{Motivation.} In several engineering applications such as aircraft tracking, attack (unknown input)/fault detection and mitigation in cyber-physical systems and urban transportation \cite{liu2011robust,yong2016tcps,yong2018simultaneous}, algorithms for unknown input reconstruction and state estimation have \yongs{become increasingly} 
indispensable and crucial \yongs{to ensure their smooth and safe operation.} 
Specifically, in safety-critical bounded-error systems, set/interval membership methods have been applied to guarantee hard accuracy bounds.
Further, in adversarial settings with potentially strategic unknown inputs, it is 
critical and desirable to simultaneously derive compatible estimates of states and unknown inputs, without assuming \moha{any \emph{a priori} known bounds/intervals for the input signals}.

\emph{Literature review.} 
Interval observer design has been extensively studied in the literature \cite{jaulin2002nonlinear,kieffer2004guaranteed,moisan2007near,raissi2010interval,raissi2011interval,mazenc2011interval,mazenc2013robust,wang2015interval,efimov2013interval,zheng2016design}. However, relatively restrictive assumptions about the existence of certain system properties were imposed to guarantee the \yongs{applicability} 
of the proposed approaches, such as 
cooperativeness \cite{raissi2010interval}, linear time-invariant (LTI) dynamics \cite{mazenc2011interval}, linear parameter-varying (LPV) dynamics that admits a diagonal Lyapunov function \cite{wang2015interval}, monotone dynamics \cite{moisan2007near}, \yongs{and} Metzler and/or Hurwitz partial linearization of nonlinearities \cite{ellero2019unknown,raissi2011interval,mazenc2013robust}. 
An $L_2/L_{\infty}$ unknown input interval observer \mohamm{design} \yongs{for continuous-time LPV systems} is studied in \cite{ellero2019unknown}. 
However, 
\yongs{this approach is not applicable for general discrete-time nonlinear }
\yongz{dynamics and, moreover, 
the  considered system do not include 
unknown inputs that} 
affect the output equation. 

Leveraging \emph{bounding functions}, the design of interval observers for a class of continuous-time nonlinear systems without unknown inputs has been addressed 
in \cite{efimov2013interval}. However, no necessary and/or sufficient conditions for the existence of bounding functions or how to compute them have been discussed. Moreover, to conclude stability, somewhat restrictive assumptions on the nonlinear dynamics have been imposed. \yongs{On the other hand,} the authors in \cite{zheng2016design} studied interval state estimation 
by extracting a known nominal observable subsystem 
and designing the observer for the transformed system, but without 
guaranteeing 
that the derived functional bounds 
are bounded sequences. Moreover, the provided conditions for the existence and stability of the observer are not \emph{constructive}.

The problem of simultaneously designing state and unknown input set-valued observers \syong{(with sets represented by $\ell_2$-norm \mohk{hyper}balls)}  has been studied in our prior works for LTI \cite{yong2018simultaneous}, LPV \cite{khajenejad2019simultaneous}, switched linear 
\cite{khajenejadasimultaneous} and nonlinear \cite{khajenejad2020nonlinear} systems with bounded-norm noise. Further, our recent work \cite{khajenejad2020simultaneous} considered the design of state and unknown input interval observers 
for nonlinear systems but with the assumption of a full-rank direct feedthrough matrix. 

 \emph{Contributions.} 
\yongs{By leveraging a combination of nonlinear decomposition mappings \cite{yang2019sufficient,coogan2015efficient} and affine abstraction (bounding) functions\yongs{\cite{singh2018mesh}},  we design 
an} observer that \emph{simultaneously} returns interval-valued estimates of states and unknown inputs for \moh{a broad range of nonlinear} systems \cite{yang2019tight}, 
\yongs{in contrast to} \yong{existing} 
  interval observers in the literature that to the best of our knowledge, only return either state \cite{jaulin2002nonlinear,kieffer2004guaranteed,moisan2007near,raissi2010interval,raissi2011interval,mazenc2011interval,mazenc2013robust,wang2015interval,efimov2013interval,zheng2016design} or input \cite{ellero2019unknown} estimates.
  Moreover, we consider  
 arbitrary \yong{unknown} input signals 
 \yongs{with no assumptions of} \moham{\emph{a priori} known bounds/intervals}, being stochastic with zero mean (as is often assumed for noise) or bounded. Further, we relax the assumption of a full-rank feedthrough matrix 
 in \cite{khajenejad2020simultaneous}, 
 \yongs{and extend the observer design \mohamm{to the systems with \syong{(possibly)} \yongz{rank-deficient} 
 feedthrough matrices}}.

In addition, we derive \mohamm{necessary and} sufficient \mohamm{rank} conditions for the existence of our observer \yongs{that can be viewed as structural properties of the nonlinear systems, as an extension of the rank condition that is typically assumed in linear state and input estimation, e.g., \cite{liu2011robust,yong2016tcps,yong2018simultaneous}.}
We also provide 
\moh{several} sufficient conditions \moh{in the form of Linear Matrix Inequalities (LMI)} for the stability of our designed observer \yongs{(i.e., the uniform boundedness of the sequence of estimate interval widths)}.
 In addition, we show that given the state intervals and specific decomposition functions, our input interval estimates are \emph{tight} and further provide upper \yongs{bound} sequences for the interval widths and derive \yongs{sufficient conditions for their convergence and their corresponding steady-state values.} 
\section{Preliminaries}
{\emph{{Notation}.}} $\mathbb{R}^n$ denotes the $n$-dimensional Euclidean space and $\mathbb{R}_{+{+}}$ positive real numbers. 
For $v,w \in \mathbb{R}^n$ and $M \in \mathbb{R}^{p \times q}$, $\|v\|\triangleq \sqrt{v^\top v}$ and $\|M\|$ denote their (induced) 2-norm, \yongn{and} 
$v \leq w$ is an element-wise inequality. 
Moreover, the transpose, 
Moore-Penrose pseudoinverse{, $(i,j)$-th element}, \mohk{the \syong{largest} 
eigenvalue} 
and rank of $M$ are given by $M^\top$, 
$M^\dagger${, $M_{i,j}$}, \mohk{$\lambda_{\max}(M)$} 
and ${\rm rk}(M)$, \mohk{respectively}.  \yongs{$M_{r:s}$} is a sub-matrix of $M$, consisting of its $r$-th through $s$-th rows. We call $M$ a non-negative matrix, i.e., $M \geq 0$, if $M_{i,j} \geq 0, \forall i \in \{1\dots p\},\forall j \in \{1 \dots q \}$. We also define \mohamm{$M^+\triangleq \max(M,0_{p \times q}),M^{-} \triangleq M^+-M$} 
\moha{and $|M| \triangleq M^++M^{-}$}. 
For a symmetric matrix $S$, $S \succ 0$ and $S \prec 0$ ($S \succeq 0$ and $S \preceq 0$) are
 positive and negative (semi-)definite, respectively.
 
Next, we introduce some useful definitions and related results. 

\begin{defn}[Interval, Maximal and Minimal Elements, Interval Width]\label{defn:interval}
{An (multi-dimensional) interval {$\mathcal{I}  \subset 
\mathbb{R}^n$} is the set of all real vectors $x \in \mathbb{R}^n$ that satisfies $\underline{s} \le x \le \overline{s}$, where $\underline{s}$, $\overline{s}$ and $\|\overline{s}-\underline{s}\|$ are called minimal vector, maximal vector and width of $\mathcal{I}$, respectively}.
\end{defn}
\moha{Next, we will briefly restate our previous result in \cite{singh2018mesh}, tailoring it \yongs{specifically for intervals to help with computing} 
affine bounding functions for our vector fields.}
\begin{prop}\cite[Affine Abstraction]{singh2018mesh}\label{prop:affine abstractions}
Consider the vector field $f(.):\mathcal{B} \subset \mathbb{R}^n \to \mathbb{R}^m$, where $\mathcal{B}$ is an interval 
with $\overline{x},\underline{x},\mathcal{V}_{\mathcal{B}}$ being its maximal, minimal and set of vertices, respectively. Suppose $\overline{A}_{\mathcal{B}},\underline{A}_{\mathcal{B}}, \overline{e}_{\mathcal{B}}, \underline{e}_{\mathcal{B}},\theta_{\mathcal{B}}$ \moham{is a solution of} the following \yongs{linear program (LP):}   
\begin{align} \label{eq:abstraction}
&\min\limits_{\theta,\overline{A},\underline{A},\overline{e},\underline{e}} {\theta} \\
\nonumber & \quad \quad s.t \ \underline{A} {x}_{s}+\underline{e}+\sigma \leq f({x}_{s}) \leq \overline{A} {x}_{s}+\overline{e}-\sigma, \\
\nonumber &\quad \quad \quad \ (\overline{A}-\underline{A}) {x}_{s}+\overline{e}-\underline{e}-2\sigma \leq \theta \mathbf{1}_m , \ \forall x_s \in \mathcal{V}_{\mathcal{B}},
\end{align}  
where $\mathbf{1}_m \in \mathbb{R}^m$ is a vector of ones and $\sigma$ can be computed via \yongs{\cite[Proposition 1]{singh2018mesh}} 
\moha{ for different function classes}. Then, $\underline{A} {x}+\underline{e} \leq f(x) \leq  \overline{A} {x}+\overline{e}, \forall x \in \mathcal{B}$. We call $\overline{A},\underline{A}$ upper and lower affine abstraction slopes of function $f(.)$ on 
$\mathcal{B}$. 
\end{prop}
\begin{prop}\cite[Lemma 1]{efimov2013interval}\label{prop:bounding}
Let $A \in \mathbb{R}^{m \times n}$ and $\underline{x} \leq x \leq \overline{x} \in \mathbb{R}^n$. Then\moh{,} $A^+\underline{x}-A^{-}\overline{x} \leq Ax \leq A^+\overline{x}-A^{-}\underline{x}$. As a corollary, if $A$ is non-negative, $A\underline{x} \leq Ax \leq A\overline{x}$. 
\end{prop}
\begin{lem} \label{lem:tightness}
Suppose the assumptions in Proposition \ref{prop:bounding} hold. Then, the returned bounds for $Ax$ is tight, in the sense that $\sup\limits_{\underline{x} \leq x \leq \overline{x}}Ax=A^+\overline{x}-A^{-}\underline{x}$ and $\inf\limits_{\underline{x} \leq x \leq \overline{x}}Ax=A^+\underline{x}-A^{-}\overline{x}$, where $\sup$ and $\inf$ are considered element-wise.
\end{lem}
\begin{proof}
For $j \hspace{-.1cm}\in \hspace{-.1cm} \{1\dots m\}$, consider the problem of $\overline{s}_j=\max\limits_{\underline{x} \leq x \leq \overline{x}}[Ax]_j$, where $[Ax]_j = \sum_{i=1}^n A_{j,i}x_i$ is the $j$-th argument of the vector $Ax$. Obviously, the solutions of this program are $x^*_i=\overline{x}_i$ if $A_{i,j} \geq 0$, and $x^*_i=-\underline{x}_i$ if $A_{i,j} < 0, \forall i \in \{1\dots n \}$. Hence $\overline{s}_j =[A]^+_j\overline{x}-[A]^{-}_j \underline{x}$, where $[A]_j$ is the $j$-th row of $A$. Similarly, $\underline{s}_j=\min_{\underline{x} \leq x \leq \overline{x}}[Ax]_j=[A]^+_j\underline{x}-[A]^{-}_j \overline{x}$. The proof is complete, since 
 $\sup_{\underline{x} \leq x \leq \overline{x}}Ax=[ \overline{s}_1  \dots \overline{s}_m ]^\top$ (similar for $\inf$).
\end{proof}
\begin{defn}[Lipschitz Continuity]\label{defn:lip}
function $f(\cdot):\mathbb{R}^n \rightarrow \mathbb{R}^m$ is $L_f$-Lipschitz continuous on $ \mathbb{R}^n$, if $\exists L_f \in \mathbb{R}_{+{+}}$, such that $\|f(x_1)-f(x_2)\| \leq L_f \|x_1-x_2\|$, $ \forall x_1,x_2 \in \mathbb{R}^n$.
\end{defn}
\begin{defn}[Mixed-Monotone Mappings and Decomposition Functions]\cite[Definition 4]{yang2019sufficient}\label{defn:mixed-monotone}
A mapping $f:\mathcal{X} \subseteq \mathbb{R}^n \rightarrow \mathcal{T} \subseteq \mathbb{R}^m$ is mixed monotone if there exists a decomposition function $f_d:\mathcal{X} \times \mathcal{X} \rightarrow \mathcal{T}$ \mohamm{\yongz{that} is monotonically increasing and decreasing in its first and second arguments, respectively, and \yongz{satisfies} $f_d(x,x)=f(x), \forall x \in \mathcal{X}$}.
\end{defn}
\begin{prop}\mohk{\cite[Theorem 1]{coogan2015efficient},\cite{moisan2005interval}} \label{prop:embedding}
Let $f:\mathcal{X} \subseteq \mathbb{R}^n \rightarrow \mathcal{T} \subseteq \mathbb{R}^m$ be a mixed monotone mapping with decomposition function $f_d:\mathcal{X} \times \mathcal{X} \rightarrow \mathcal{T}$ and $\underline{x} \leq x \leq \overline{x}$, where $\underline{x},x,\overline{x} \in \mathcal{X}$. Then $f_d(\underline{x},\overline{x}) \leq f(x) \leq f_d(\overline{x},\underline{x})$.
\end{prop}
\syong{Note that decomposition functions for nonlinear functions are not unique, and several different ones have been proposed in the literature, e.g., \cite{mazenc2014interval,yang2019sufficient,moisan2005interval,coogan2015efficient,yang2019tight}. Although any decomposition function can be used in conjunction with our proposed interval observer, we will adopt the specific decomposition function given in \cite[Theorem 2]{yang2019sufficient}, since it allows us to derive a \emph{Lipschitz-like} property (in Lemma \ref{lem:lip-dec}) that further enables us to derive sufficient conditions for the stability of the proposed observer in \eqref{eq:xup}--\eqref{eq:hlow}. 

Thus, we now briefly describe the decomposition function given in \cite[Theorem 2]{yang2019sufficient} that we will adopt in this paper: 
}

If a vector field $q=\begin{bmatrix} h^\top_1 & \dots & q^\top_n \end{bmatrix}^\top:X \subseteq \mathbb{R}^n \rightarrow \mathbb{R}^m$ is differentiable and its partial derivatives are bounded with known bounds, i.e., $\frac{\partial q_i}{\partial x_j} \in (a^q_{i,j},b^q_{i,j}), \forall x \in X \in \mathbb{R}^n$, where $a^q_{i,j},b^q_{i,j} \in \overline{\mathbb{R}}$, then $h$ is mixed monotone with a decomposition function $q_d=\begin{bmatrix} q^\top_{d1} & \dots & q^\top_{di} & \dots q^\top_{dn} \end{bmatrix}^\top$, where $q_{di}(x,y)=q_{i}(z)+(\alpha^q_i-\beta^q_i)^\top (x-y), \forall i \in \{1,\dots,n\}$, and $z,\alpha^q_i,\beta^h_i \in \mathbb{R}^n$ can be computed in terms of $x, y, a^q_{i,j}, b^q_{i,j}$ as given in \cite[(10)--(13)]{yang2019sufficient}. Consequently, for $x=[x_1\dots x_j \dots x_n]^\top$, $y=[y_1\dots y_j \dots y_n]^\top$, we have 
\begin{align}\label{eq:decompconstruct}
q_{d}(x,y)&=q(z)+C_q(x-y),
\end{align}
where $C_q \hspace{-.1cm}  \triangleq \hspace{-.1cm}  \begin{bmatrix} [\alpha^q_1-\beta^q_1] & \hspace{-0.2cm}\dots \hspace{-0.2cm} & [\alpha^q_i-\beta^q_i] & \dots [\alpha^q_m-\beta^q_m] \end{bmatrix}^\top \hspace{-.2cm} \in \hspace{-.1cm}  \mathbb{R}^{m \times n}$, 
with 
$\alpha^q_i,\beta^q_i$ given in \cite[(10)--(13)]{yang2019sufficient}, $z=[z_1 \dots z_j \dots z_m]^\top$ and $z_j=x_j$ or $y_j$ (dependent on the case, cf. \cite[Theorem 1 and (10)--(13)]{yang2019sufficient} for details). Moreover, if \yongs{exact values of} $a_{i,j}, b_{i,j}$ 
are unknown, 
their approximations can be obtained using Proposition \ref{prop:affine abstractions} \yongs{with the slopes set to 0.} 
\begin{cor}\label{cor:bounding}
As a direct implication of Propositions \ref{prop:affine abstractions}--\ref{prop:embedding}, for any Lipschitz mixed-monotone vector-field $q(.):\mathbb{R}^n \to \mathbb{R}^m$, with a decomposition function $q_d(.,.)$, we can find upper and lower vectors $\overline{q},\underline{q}$ such that $\underline{q} \leq q(x) \leq \overline{q}, \forall x \in [\underline{x},\overline{x}]$, and 
\begin{align}
\nonumber &\underline{q} = \max(q_d(\underline{x},\overline{x}),\hat{\underline{q}}), \quad \overline{q} = \min(q_d(\overline{x},\underline{x}),\hat{\overline{q}}), \\
 \nonumber &\hat{\underline{q}} =(\underline{A}^q)^+\underline{x}-(\underline{A}^q)^{-}\overline{x}+\underline{e}^q,\hat{\overline{q}} = (\overline{A}^q)^+\overline{x}\hspace{-.05cm}-\hspace{-.05cm}(\overline{A}^q)^{-}\underline{x}\hspace{-.05cm}+\hspace{-.05cm}\overline{e}^q \hspace{-.05cm}, 
\end{align}
where $(\overline{A}^q,\underline{A}^q,\overline{e}^q,\underline{e}^q)$ \moha{is a solution of \eqref{eq:abstraction}} \yongs{for the function $q$}. 
\end{cor}
Finally, we \mohamm{restate} a Lipschitz-like property for the bounding functions 
in Corollary \ref{cor:bounding}, which \mohamm{we derived in \cite{khajenejad2020simultaneous}} and will be used later for \yongs{determining observer stability}. 
\vspace{-.2cm}
\begin{lem}\cite[Lemma 1]{khajenejad2020simultaneous} \label{lem:lip-dec}
Let $q(.):[\underline{x},\overline{x}] \subset \mathbb{R}^n \to \mathbb{R}^m$ be the Lipschitz mixed-monotone vector-field in Corollary \ref{cor:bounding}, with its decomposition function $q_d(.,.)$ constructed using \eqref{eq:decompconstruct}. Then, $\| \overline{q}-\underline{q}\| \leq \|q_d(\overline{x},\underline{x})-q_d(\underline{x}_,\overline{x})\| \leq L_{q_d}\|\overline{x}-\underline{x}\|$, where $L_{q_d} \triangleq L_q+2\|C_q\|$, with $C_q$ given in \eqref{eq:decompconstruct}.
\end{lem}

\vspace{-0.3cm}
\section{Problem Formulation} \label{sec:Problem}
\vspace{-0.1cm}
\noindent\textbf{\emph{System Assumptions.}} 
Consider the nonlinear discrete-time system with unknown inputs and bounded noise 
\begin{align} \label{eq:system}
\begin{array}{ll}
x_{k+1}&=f(x_k)+B u_k+G d_k + w_k,\\
y_k&=g(x_k) +D u_k + H d_k + v_k, \end{array}
\end{align}
where at time $k \in \mathbb{N}$, $x_k \in \mathbb{R}^n$, $u_k \in \mathbb{R}^m$, $d_k \in \mathbb{R}^p$ and $y_k \in \mathbb{R}^l$ are the state, 
 a known input, an unknown input and the measurement vectors, correspondingly. The process and measurement noise signals $w_k \in \mathbb{R}^n$ and $v_k \in \mathbb{R}^l$ are assumed to be bounded \mohamm{and satisfy} $\underline{w} \leq w_k \leq \overline{w}$, $\underline{v} \leq v_k \leq \overline{v}$, \mohamm{with} the known lower and upper bounds, $\underline{w}$, $\overline{w}$ and $\underline{v}$, $\overline{v}$, respectively. 
We also assume $\underline{x}_0 \leq x_0 \leq \overline{x}_0$, \mohamm{where} $\underline{x}_0$ and $\overline{x}_0$, {are} \mohamm{available} lower and upper bounds for the initial state $ x_0 $. The vector fields $f(\cdot):\mathbb{R}^n \rightarrow \mathbb{R}^n$, $g(\cdot):\mathbb{R}^n \rightarrow \mathbb{R}^l$ and matrices $B$, $D$, $G$ and $H$ are known 
and of appropriate dimensions, where $G$ and $H$ 
encoding the \emph{locations} through which the unknown input \yongn{(or attack)} signal can affect the system dynamics and measurements. Note that no assumption is made on $H$ to be either the zero matrix or to have full column rank when there is direct feedthrough (in contrast \yongs{to} \cite{khajenejad2020simultaneous}). \mohamm{Without loss of generality, we assume that $\begin{bmatrix} G^\top & H^\top \end{bmatrix}^\top$ is full-rank}.  
Moreover, \mohamm{the following, which is satisfied for a broad range of nonlinearities \cite{yang2019tight}, is assumed}. 

\vspace{-0.1cm}
\begin{assumption}\label{assumption:mix-lip}
Vector fields $f(\cdot)$ and $g(\cdot)$ are mixed-monotone with decomposition functions $f_d(\cdot,\cdot)$ and $g_d(\cdot,\cdot)$ and $L_f$-Lipschitz and $L_g$-Lipschitz continuous, respectively. 
\end{assumption} \vspace{-0.1cm}

\noindent \textbf{\emph{Unknown Input (or Attack) Signal Assumptions.}} 
The unknown inputs $d_k$ are not constrained to follow any model nor to be a signal of any type (random or strategic), hence no prior `useful' knowledge of the dynamics of $d_k$ is available (independent of $\{d_\ell\}$ $\forall k\neq \ell$, $\{w_\ell\}$ and $\{v_\ell\}$ $ \forall  \ell$). We also do not assume that $d_k$ is bounded or has known bounds and thus, $d_k$ is suitable for representing adversarial 
attack signals.

The 
observer design problem 
 can be stated as follows:
\begin{problem}\label{prob:SISIO}
Given a nonlinear discrete-time system with unknown inputs and bounded noise \eqref{eq:system}, design a stable observer that simultaneously finds bounded intervals 
of compatible states and unknown inputs. 
\end{problem}

\section{General Simultaneous Input and State Interval Observers (GSISIO)} \label{sec:observer}
\subsection{Interval Observer Design} \label{sec:obsv}
We consider a recursive \mohamm{two}-step interval-valued observer design, composed of a \emph{state propagation} step, which propagates the previous time state estimates through the state equation to find propagated intervals, 
and an \emph{unknown input estimation} step, which computes the input intervals using state intervals and observation. 
Note that we are constrained with obtaining a one-step delayed estimate of \mohamm{the unknown input signal}, 
because in contrast with \cite{khajenejad2020simultaneous}, the matrix $H$ is not necessarily full-rank, and hence $d_k$ cannot be estimated from the current measurement, $y_k$. However, in Lemma \ref{lem:partial-estimate}, we will discuss a way of obtaining the current estimate of \emph{a component of} the input signal. 

Considering the computational complexity of optimal observers \cite{milanese1991optimal}, as well as nice properties of interval sets \cite{ellero2019unknown}, we consider set estimates of the form: 
\begin{align*}
\mathcal{I}^x_{k}\hspace{-.1cm}=\hspace{-.1cm}\{x \hspace{-.1cm}\in\hspace{-.1cm} \mathbb{R}^n: \underline{x}_{k} \hspace{-.1cm}\leq \hspace{-.1cm} x \leq\hspace{-.1cm} \overline{x}_{k}\},
\mathcal{I}^d_{k-1}\hspace{-.1cm}=\hspace{-.1cm}\{d \hspace{-.1cm}\in\hspace{-.1cm} \mathbb{R}^p: \underline{d}_{k-1} \hspace{-.1cm}\leq \hspace{-.1cm}d \hspace{-.1cm}\leq \hspace{-.1cm}\overline{d}_{k-1} \},
\end{align*}
i.e., we restrict the estimation errors to \yongs{be} closed intervals. 
In this case, the observer design problem boils down to finding 
$\underline{x}_{k}$, $\overline{x}_{k}$, $\underline{d}_{k-1}$ and $\overline{d}_{k-1}$. 
Our interval observer can be defined at each time step $k \geq 1$ as follows (with known $\underline{x}_{0}$ and $\overline{x}_0$ such that $\underline{x}_{0} \leq x_0 \leq \overline{x}_0$):\\[0.2cm]
\noindent \textbf{\emph{State Propagation}}: \vspace{-0.1cm}
\begin{align} \label{eq:xup}
\begin{array}{ll}
  & \hspace{-.5cm}\begin{bmatrix} \overline{x}_{k} & \underline{x}_{k} \end{bmatrix}^\top \hspace{-.1cm}= \hspace{-.05cm} M_f \begin{bmatrix} \overline{f}^\top_k & \underline{f}^\top_k \end{bmatrix}^\top\hspace{-.1cm}+ \hspace{-.05cm}M_g\begin{bmatrix} \overline{g}^\top_k &  \underline{g}^\top_k \end{bmatrix}^\top\hspace{-.05cm}+ 
  \\
 & \hspace{-.15cm}M_v \begin{bmatrix} \overline{v}^\top & \underline{v}^\top \end{bmatrix}^\top\hspace{-.2cm}+\hspace{-.05cm}M_w \begin{bmatrix} \overline{w}^\top & \underline{w}^\top \end{bmatrix}^\top \hspace{-.15cm}+\hspace{-.1cm}M_y y_{k-1}\hspace{-.1cm}+\hspace{-.1cm}M_u u_{k-1};  
\end{array}
\end{align}
\noindent \textbf{\emph{Unknown Input Estimation}}: \vspace{-0.1cm}
\begin{align}
 \overline{d}_{k-1}=N_{11}\overline{h}_{k}+ N_{12} \underline{h}_{k}, \quad  \underline{d}_{k-1}= N_{21} \overline{h}_{k}+ N_{22} \underline{h}_{k} , \label{eq:dup}
\end{align}
where $\forall q \in \{f,g \}$, $\overline{q}_k$ and $\underline{q}_k$ are upper and lower vector values for the function $q(.)$ on the interval $[\underline{x}_{k-1},\overline{x}_{k-1}]$, which can be recursively computed using Corollary \ref{cor:bounding}. Moreover,

\hspace{0.1cm}\phantom{a}\hspace{0.1cm}\vspace{-.3cm}
\begin{small}
\begin{align}
&\overline{h}_{k} \hspace{-.1cm}=\hspace{-.1cm} \begin{bmatrix} \overline{x}^\top_k & \hspace{-.1cm}y^\top_{k-1} \end{bmatrix}\hspace{-.1cm}^\top\hspace{-.15cm}-\hspace{-.1cm}\begin{bmatrix}   \underline{f}^\top_k &   \hspace{-.1cm}\underline{g}^\top_k  \end{bmatrix}\hspace{-.1cm}^\top\hspace{-.15cm}- \hspace{-.1cm}\begin{bmatrix} B^\top & \hspace{-.1cm}D^\top \end{bmatrix}\hspace{-.1cm}^\top \hspace{-.1cm}u_{k-1}\hspace{-.1cm}-\hspace{-.1cm}\begin{bmatrix} \underline{w}^\top & \hspace{-.1cm}\underline{v}^\top \end{bmatrix}\hspace{-.1cm}^\top,  \label{eq:hup} \\
 &\underline{h}_{k} \hspace{-.1cm}=\hspace{-.1cm} \begin{bmatrix} \underline{x}^\top_k & \hspace{-.1cm}y^\top_{k-1} \end{bmatrix}\hspace{-.1cm}^\top\hspace{-.15cm}-\hspace{-.1cm}\begin{bmatrix}   \overline{f}^\top_k &   \hspace{-.1cm}\overline{g}^\top_k  \end{bmatrix}\hspace{-.1cm}^\top\hspace{-.15cm}- \hspace{-.1cm}\begin{bmatrix} B^\top & \hspace{-.1cm}D^\top \end{bmatrix}\hspace{-.1cm}^\top \hspace{-.1cm}u_{k-1}\hspace{-.1cm}-\hspace{-.1cm}\begin{bmatrix} \overline{w}^\top & \hspace{-.1cm}\overline{v}^\top \end{bmatrix}\hspace{-.1cm}^\top. \label{eq:hlow} 
\end{align}
\end{small}\vspace{-.4cm}

  Finally, 
  $M_s$, $N_{nm}$, $\forall s \in \{f,g,u,w,v,y\},n,m \in \{1,2\}$, 
  are to-be-designed observer gains 
  \yongz{in order to} achieve desirable observer properties. Algorithm \ref{algorithm1} summarizes GSISIO. 
\mohk{
\begin{algorithm}[t] \small
\caption{GSISIO}\label{algorithm1}
\begin{algorithmic}[1]
\Function{$GSISIO$}{$f(\cdot),g(\cdot),B,G,D,H,\overline{w},\underline{w},\overline{v},\underline{v},\overline{x}_0,\underline{x}_0$}
		\State Initialize: $\text{maximal}(\mathcal{I}^x_0)=\overline{x}_0$; $\text{minimal}(\mathcal{I}^x_0)=\underline{x}_0$;
		\Statex  Compute \hspace{-0.05cm}$M_s, \hspace{-0.05cm}N_{ij},\hspace{-0.1cm} \forall s \hspace{-0.1cm}\in\hspace{-0.1cm} \{f,g,u,v,w\},i,j \hspace{-0.1cm}\in\hspace{-0.1cm} \{1,2\}$\hspace{-0.1cm} via \hspace{-0.05cm}Theorem \ref{thm:existence};
		\For {$k =1$ to $\overline{K}$}

    	\LineComment{\hspace{0.3cm}Estimation of ${x}_{k}$ (to compute $\mathcal{I}^x_{k}\hspace{-.1cm}=[\underline{x},\overline{x}_{k}]$):}
		\Statex \hspace{0.6cm} Compute $\overline{x}_{k},\underline{x}_{k}$ via\hspace{-0.05cm} \eqref{eq:xup}; Compute 
		$\delta^x_{k}$ through Lemma \ref{lem:convergence};
		
				\LineComment{\hspace{0.3cm}Estimation of ${d}_{k-1}$ (to compute $\mathcal{I}^d_{k-1}\hspace{-.1cm}=[ \underline{d}_{k-1},\overline{d}_{k-1}]$):}
		 \Statex \hspace{0.6cm} 
		 Compute $\overline{d}_{k-1},\underline{d}_{k-1}$ and  $\delta^d_{k-1}$ via \eqref{eq:dup}--\eqref{eq:hlow} and Lemma \ref{lem:convergence};
		 \Statex \hspace{0.6cm} \Return {$\mathcal{I}^x_k,\mathcal{I}^d_{k-1},\delta^x_{k},\delta^d_{k-1}$};
		\EndFor
\EndFunction
		\end{algorithmic}

\end{algorithm}   }  
\vspace{-0.1cm} 
\vspace{-0.4cm}
\subsection{Observer Design}\vspace{-0.05cm}
The objective of this section is to design observer gains such that the GSISIO returns \emph{correct} and \emph{tight} intervals. We first define these properties through the following definitions. 
\begin{defn}[Correctness \moh{(Framer Property \cite{mazenc2013robust})}] \label{defn:correctness}
\yongn{Given an initial interval $x_0 \in [\underline{x}_0,\overline{x}_0]$, the 
GSISIO observer returns correct interval estimates, 
if the true states and unknown inputs of the system \eqref{eq:system} are 
within the estimated intervals \eqref{eq:xup}--\eqref{eq:dup} for all times}. 
If the observer is correct, we call $\{\overline{x}_k,\underline{x}_k,\overline{d}_{k-1},\underline{d}_{k-1}\}_{k=1}^{\infty}$ 
(the state and input) framers.  
\end{defn}
\begin{defn}[Tightness of Input Estimates] \label{defn:tightness}
The input interval estimates 
are tight, if at each time step $k$, given the state estimate, the input framers $\overline{d}_{k-1},\underline{d}_{k-1}$, coincide with supremum and infimum values of the set of compatible inputs. 
\end{defn} \vspace{-0.05cm}
\mohamm{First, the tightness} of the input framers \mohamm{\yongz{is} addressed}. 
\begin{lem}[Correctness and Tightness of Input Estimates]\label{lem:framerd}
Consider the system \eqref{eq:system} along with the GSISIO in \eqref{eq:xup}--\eqref{eq:dup}. Suppose that Assumption \ref{assumption:mix-lip} holds. Let $J \triangleq (\begin{bmatrix} G^\top & H^\top \end{bmatrix}^\top)^\dagger$, $N_{11}=N_{22}=J^+$ and $N_{12}=N_{21}=-J^{-}$. Then, given any pair of state framer sequences $\{\overline{x}_k,\underline{x}_k\}_{i=0}^{\infty}$, the input interval estimates given in \eqref{eq:dup} are correct and tight. 
\end{lem}\vspace{-0.05cm}
\begin{proof}
Augmenting the state and output equations in \eqref{eq:system} and \yongs{from} Corollary \ref{cor:bounding}, \yongs{we obtain} 
$\underline{h}_k \leq \begin{bmatrix} G^\top & H^\top \end{bmatrix}^\top d_{k-1} \leq \overline{h}_k $, with $\underline{h}_k,\overline{h}_k$ defined in \eqref{eq:hup},\eqref{eq:hlow}. Then, the input framers in \eqref{eq:dup} can be obtained by 
using Propositions \ref{prop:affine abstractions}--\ref{prop:embedding} and considering the fact that $\begin{bmatrix} G^\top & H^\top \end{bmatrix}^\top $ is full rank. Finally, tightness is implied by Lemma \ref{lem:tightness} (where the $A$ matrix equals $J$). 
\end{proof}
\mohamm{Next, we address
the existence 
of correct framers}. 
\vspace{-0.05cm}  
\begin{thm}[Existence of Correct Framers] \label{thm:existence}
Consider the system \eqref{eq:system} \mohamm{and the} 
GSISIO introduced in \eqref{eq:xup}-\eqref{eq:dup}. Suppose all the assumptions in Lemma \ref{lem:framerd} hold and 
the observer gains are chosen \mohk{as follows:} 
\begin{align}
\nonumber \forall s &\in \{f,g,u,w,v,y \}: M_s= A_x^\dagger A_s,A_v=A_g,\\
\nonumber A_u &\triangleq  \begin{bmatrix} F^\top & F^\top \end{bmatrix}^\top,A_w \hspace{-.1cm}= A_f,A_g\triangleq  \begin{bmatrix}L_2 & -K_2 \\ -K_2 & L_2\end{bmatrix}, \\
\nonumber A_x&\triangleq  \begin{bmatrix}I-K_1 & L_1 \\ L_1 & I-K_1\end{bmatrix},A_f \triangleq \begin{bmatrix}I+L_1 & -K_1 \\ -K_1 & I+L_1\end{bmatrix},\\
\nonumber  L &\triangleq G^{-}J^{+}\hspace{-.1cm}+G^+J^{-},K \triangleq G^{-}J^{-}+G^+J^+, \\
\nonumber  K_1& \triangleq K_{1:n},K_2 \triangleq K_{n+1:n+l},L_1 \triangleq L_{1:n},L_2 \triangleq L_{n+1:n+l},\\
\nonumber F &\triangleq (I +L_1- K_1)B + (L_2-K_2)D.
\end{align}
Then, at each 
time step, the GSISIO returns finite and correct framers if  \mohamm{and only if}
\begin{align} \label{eq:existence}
{\rm rk}(I-K_1-L_1)={\rm rk}(I-K_2+L_2)=n.
\end{align} 
\end{thm}
\begin{proof}
From the state equation in \ref{eq:system}, Corollary \ref{cor:bounding} and Proposition \ref{cor:bounding}, we have $\underline{x}_k \leq x_k \leq \overline{x}_k$, where, $\underline{x}_k =\underline{f}_k+Bu_{k-1}+\underline{w}+ G^+\underline{d}_{k-1}-G^{-}\overline{d}_{k-1}$, $\overline{x}_k =\overline{f}_k+Bu_{k-1}+\overline{w}+ G^+\overline{d}_{k-1}-G^{-}\underline{d}_{k-1}$, \mohamm{which. in addition to \eqref{eq:dup}--\eqref{eq:hlow}, results in} 
the following system of linear equations: 

\vspace*{-0.4cm} \begin{align} 
\begin{array}{ll}
 &\hspace{-.33cm}A_x\hspace{-.1cm} \begin{bmatrix} \overline{x}^{\top}_{k} & \underline{x}^{\top}_{k} \end{bmatrix}^\top \hspace{-.3cm}= \hspace{-.1cm}A_f \hspace{-.1cm}\begin{bmatrix} \overline{f}^{\top}_k &  \underline{f}^{\top}_k \end{bmatrix}^\top\hspace{-.2cm}+ \hspace{-.1cm}A_g\hspace{-.1cm}\begin{bmatrix} \overline{g}^{\top}_k &  \underline{g}^{\top}_k \end{bmatrix}^\top\hspace{-.2cm}+\hspace{-.1cm}A_u u_{k-1} \label{eq:framerx}\\
  &\hspace{0cm}+A_w \begin{bmatrix} \overline{w}^\top & \underline{w}^\top \end{bmatrix}^\top\hspace{-.1cm}+\hspace{-.1cm}A_v \begin{bmatrix} \overline{v}^\top &\underline{v}^\top \end{bmatrix}^\top\hspace{-.1cm}+\hspace{-.1cm}A_y y_{k-1} \moha{\triangleq p_k},
  \end{array}
\end{align}
with 
$\overline{q}_k,\underline{q}_k, \forall q \in \{f,g\}$ obtained from Corollary \ref{cor:bounding} with the corresponding interval $[\underline{x}_{k-1},\overline{x}_{k-1}]$. By \cite{james1978generalised}, the \moha{set of all solutions} of \eqref{eq:framerx} lies in \mohamm{the interval $[\underline{x}_k,\overline{x_k}]$, where} 
\begin{align} \label{eq:propsol1}
 \overline{x}_{k}=\overline{x}^{f}_k+\mu {r}, \quad \underline{x}_{k} = \underline{x}^{f}_k-\mu {r},
\end{align}   
\mohamm{with} $\mu$ being a large positive number (infinity), $\overline{x}^{f}_k\triangleq(A^\dagger_x \moha{p_k})_{1:n},\underline{x}^{f}_k\triangleq(A^\dagger_x \moha{p_k})_{n+1:2n}$, 
and $r_i=0$ if the $i$-th row of $A$ is zero and $r_i=1$ otherwise. 
Finally, the finiteness of $\overline{x}_k,\underline{x}_k$ \yongs{reduces} to $r=0$, which is \mohamm{\emph{equivalent}} to the rank conditions in \eqref{eq:existence} by the fact that $A_x$ is a block real centro-Hermitian matrix by its definition and \cite[Corollary 4.7]{lu2002inverses}.
\end{proof}
\mohamm{Although} we can 
only obtain a one-step delayed estimate of $d_k$ in \eqref{eq:dup}, \mohamm{finding  
an estimate for \emph{a \yongz{sub}component of $d_{k}$} at current time $k$, can be 
formalized as follows}. 
\begin{lem} \label{lem:partial-estimate}
Suppose all the assumptions in Theorem \ref{thm:existence} hold. Then, at time step $k$, \mohamm{the unknown input $d_k$ can always be decomposed into two components $d_{1,k}$ and $d_{2,k}$, where there exists matrices $T_1$ and $\Phi$ such that \syong{the framers for $d_{1,k}$ can be found as:}} 
\mohk{
\begin{align*} 
 \underline{d}_{1,k} \leq d_{1,k} \leq \overline{d}_{1,k}, 
\end{align*}
\mohamm{where} }
 \begin{align*}
 \overline{d}_{1,k} &= \Phi(z_{1,k}-T_1Du_k)+\overline{\ell}_k, \\
 \underline{d}_{1,k} &= \Phi(z_{1,k}-T_1Du_k)+\underline{\ell}_k, \\ 
 \overline{\ell}_k &\triangleq (\Phi T_1)^{-}(g_d(\overline{x}_k,\underline{x}_k)+\overline{v}) 
-(\Phi T_1)^{+}(g_d(\underline{x}_k,\overline{x}_k)+\underline{v}), \\ 
\underline{\ell}_k &\triangleq (\Phi T_1)^{-}(g_d(\underline{x}_k,\overline{x}_k)+\underline{v}) 
-(\Phi T_1)^{+}(g_d(\overline{x}_k,\underline{x}_k)+\overline{v}). 
\end{align*}
Moreover, 
$d_{2,k}$  cannot be \yongz{estimated} \yongz{at the current time $k$}. 
\end{lem}
\begin{proof}
 Let $p_{H}\triangleq {\rm rk} (H)$. Similar to \cite{yong2018simultaneous}, by applying singular value decomposition, we have $H= \begin{bmatrix}U_{1}& U_{2} \end{bmatrix} \begin{bmatrix} \Sigma & 0 \\ 0 & 0 \end{bmatrix} \begin{bmatrix} V_{1}^{\, \top} \\ V_{2}^{\, \top} \end{bmatrix}$ with $V_{1}\hspace{-.1cm} \in \hspace{-.1cm}\mathbb{R}^{p \times p_{H}}$, $V_{2} \hspace{-.1cm}\in\hspace{-.1cm} \mathbb{R}^{p \times (p-p_{H})}$, $\Sigma \hspace{-.1cm}\in \hspace{-.1cm} \mathbb{R}^{p_{H} \times p_{H}}$ (a diagonal matrix of full rank),
 $U_{1} \in \mathbb{R}^{l \times p_{H}}$ and $U_{2} \in \mathbb{R}^{l \times (l-p_{H})}$. 
Then, since 
$V\triangleq [ V_{1} \ V_{2}]$ is unitary, 
$d_k =V_{1} d_{1,k}+V_{2} d_{2,k}, d_{1,k}=V_{1}^\top d_k, d_{2,k}=V_{2}^\top d_k$.  
 Moreover, by defining $T_1\triangleq U^\top_1,T_2 \triangleq U^\top_2$, the output equation can be decoupled as: 
 $z_{1,k}= g_1(x_k) + D_1u_k + v_{1,k}+\Sigma d_{1,k}$ and 
  $z_{2,k}= g_2(x_k) + D_2u_k + v_{2,k}$, where 
  $g_1(x_k)\triangleq T_1g(x_k),g_2(x_k)\triangleq T_2g(x_k)$. 
The bounds for $d_{1,k}$ \yongs{can be} obtained by applying Propositions \ref{prop:bounding} and \ref{prop:embedding} to the first equation \mohamm{and setting $\Phi=\Sigma^{-1}$}. 
Finally, since $d_{2,k}$ does not appear in any of the equations, 
it cannot be estimated at the current time.  
\end{proof}
\begin{rem} \label{rem:partestimate}
The result in Lemma \ref{lem:partial-estimate} is particularly helpful in the special case when the feedthrough matrix has full rank, so $d_k=d_{1,k}$ and hence, $d_k$ can be estimated at {current time} $k$, which is an alternative approach \yongs{to} the one in \cite{khajenejad2020simultaneous}. 
\end{rem} \vspace{-.25cm}
\subsection{\moham{Uniform} Boundedness of Estimates (Observer Stability)}
In this section, 
we 
\yongz{investigate} the stability of \yongn{GSISIO}. 
\begin{thm}[
Observer Stability] \label{thm:boundedness}
Consider the system \eqref{eq:system} and the GSISIO \eqref{eq:xup}--\eqref{eq:dup}. Suppose 
all the assumptions in Theorem \ref{thm:existence} hold and the decomposition functions $f_d, g_d$ are constructed using \eqref{eq:decompconstruct}, \mohamm{with their corresponding $L_{f_d},L_{g_d}$ given in Lemma \ref{lem:lip-dec}}. 
Then, the observer is stable, in the sense that interval width sequences $\{\|\Delta^d_{k-1}\| \triangleq \|\overline{d}_{k-1}-\underline{d}_{k-1}\|, \|\Delta^x_k\| \triangleq \|\overline{x}_{k}-\underline{x}_{k}\|\}_{k=1}^{\infty}$ 
and estimation errors $\{\|\tilde{d}_{k-1}\| \triangleq \max (\|d_{k-1}-\underline{d}_{k-1}\|,\|\overline{d}_{k-1}-{d}_{k-1}\|), \|\tilde{x}_{k}\| \triangleq \max (\|x_{k}-\underline{x}_{k}\|,\|\overline{x}_{k}-{x}_{k}\|) \}_{k=1}^{\infty}$ are uniformly bounded, if \yong{either one} of the following 
conditions hold: \vspace{0.1cm}
\renewcommand{\theenumi}{\roman{enumi}}
\begin{enumerate}[(i)]
\item ${\mathcal{L}} \triangleq L_{f_d}\|{T}_f\|+L_{g_d}\|{T}_g\| \leq 1$,  \label{item:first} \vspace{0.1cm}
\item $ {\mathcal{T}}\triangleq  \begin{bmatrix} Q & 0 &0&0&0 \\ * & {T}^\top_g {T}_g & {T}^\top_g{T}_f& {T}^\top_g{T}_f&{T}^\top_g {T}_g \\ * & * & {T}^\top_f{T}_f & {T}^\top_f{T}_f & {T}^\top_f{T}_g \\ * & * & * &0 & {T}^\top_f{T}_g\\ * & * & * & * & 0 \end{bmatrix}\preceq 0$,\label{item:third}
\item \label{item:second} There exists 
$P \succ 0$ and $\Gamma \succeq 0$ in $\mathbb{R}^{n \times n}$ 
such that \\ $\mathcal{P} \triangleq \begin{bmatrix} P+\Gamma-I & 0 & P \\ 0 & \mathcal{L}^2I-P & 0 \\ P & 0 & P \end{bmatrix} \preceq 0$, 
\end{enumerate}
\mohk{
\syong{with} 
\begin{align*}
T_f &\triangleq (I -K_1 - L_1)^{\dagger}(I -K_1 + L_1), \\
 T_g &\triangleq (I-K_1-L_1)^{\dagger}(K_2+L_2), \\
Q &\triangleq \lambda_{\max}({T}^\top_f {T}_f)L^2_{f_d}+\lambda_{\max}({T}^\top_g {T}_g)L^2_{g_d}-1, 
\end{align*}
\syong{and} $K_1, K_2,L_1, 
L_2$ are given in Theorem \ref{thm:existence}.}
\end{thm}
\begin{proof}
Let $\Delta^x_k \triangleq \overline{x}_k-\underline{x}_k$ and $\Delta x^{f}_k\triangleq \overline{x}^f_k-\underline{x}^f_k$. Then, by \eqref{eq:propsol1}, $\Delta^x_k=\Delta x^{f}_k+2\mu r$ and since \eqref{eq:existence} holds, then $\Delta^x_k=\Delta x^{f}_k$.  
 On the other hand, from \eqref{eq:framerx} and Corollary \ref{cor:bounding}, we obtain 
\begin{align}\label{eq:deltax}
\Delta x_k=\Delta x^{f}_{k} \leq \Delta \tilde{f}^x_{k-1}+\Delta z,
\end{align}
where $\Delta \tilde{f}^x_k \triangleq T_f\Delta f^x_k\hspace{-0.1cm}+T_g\Delta g^x_k$, $\Delta f^x_k \triangleq f_d(\overline{x}_k,\underline{x}_k)-f_d(\underline{x}_k,\overline{x}_k)$, $\Delta g^x_k \triangleq g_d(\overline{x}_k,\underline{x}_k)-g_d(\underline{x}_k,\overline{x}_k)$, $\Delta z \triangleq T_f\Delta w+T_g\Delta v$, $\Delta w \triangleq \overline{w}-\underline{w}$, $\Delta v \triangleq \overline{v}-\underline{v}$, $T_f \triangleq (I\hspace{-0.05cm}-\hspace{-0.05cm}K_1\hspace{-0.05cm}-\hspace{-0.05cm}L_1)^{\dagger}(I\hspace{-0.05cm}-\hspace{-0.05cm}K_1\hspace{-0.05cm}+\hspace{-0.05cm}L_1)$ and $T_g \triangleq(I\hspace{-0.05cm}-\hspace{-0.05cm}K_1\hspace{-0.05cm}-\hspace{-0.05cm}L_1)^{\dagger}(K_2\hspace{-0.05cm}+\hspace{-0.05cm}L_2)$.\\[-0.35cm]  

\noindent\textbf{
Condition (i):}
 By Assumption \ref{assumption:mix-lip}, triangle inequality and \eqref{eq:deltax}: 
$\|\Delta^x_{k}\| \leq \mathcal{L} \|\Delta^x_{k-1}\|+\|\Delta z\|$,
with $\mathcal{L} \triangleq L_{f_d}\|T_f\|+L_{g_d}\|T_g\|$ and $L_{f_d},L_{g_d}$ obtained from Lemma \ref{lem:lip-dec}. Since $\mathcal{L} \leq 1$ (by 
Condition (i)), 
the \yongs{sequence} 
$\{\|\Delta^x_k\|\}_{k=0}^{\infty}$ is uniformly bounded. Therefore, the interval width dynamics is stable.\\[-0.35cm] 

\noindent\textbf{
Condition (ii):}
 By the \emph{Comparison Lemma} \cite[Lemma 3.4]{khalil2002nonlinear} and non-negativity of $\Delta^x_k$, to show the stability of the system in \eqref{eq:deltax}, it suffices to show the uniform boundedness of $\{\Delta^s_k\}_{k=0}^{\infty}$, where
$\Delta^s_{k} = \Delta \tilde{f}^s_{k-1}+\Delta z, \ \Delta^s_{0}=\Delta^x_{0}$.
\mohamm{To do so}, consider a candidate Lyapunov function $V_k=\Delta^{s\top}_k\Delta^s_k$ 
that can be shown \mohamm{to satisfy} $\Delta V_k \triangleq V_{k+1}-V_k \leq \Psi^{\top}_k {\mathcal{T}}\Psi_k $, with $\Psi_k \triangleq \begin{bmatrix} \Delta^{s\top}_k & \Delta v^\top &\Delta w^\top & \Delta f^{s\top}_k \Delta g^{s\top}_k\end{bmatrix}^\top$ and ${\mathcal{T}}$ defined in the statement of the theorem, as follows: 
$\Delta V_k
 \leq (\lambda_{\max}({T}^\top_f {T}_f)L^2_{f_d}+\lambda_{\max}({T}^\top_g {T}_g)L^2_{g_d}-1)\Delta^{s\top}_k\Delta^s_k
 +\Delta v^\top {T}^\top_g {T}_g \Delta v+\Delta w^\top {T}^\top_f{T}_f \Delta w+2(\Delta f^{s\top}_k{T}^\top_f {T}_g \Delta g^s_k 
+\Delta f^{s\top}_k {T}^\top_f {T}_g \Delta v+\Delta f^{s\top}_k {T}^\top_f {T}_f\Delta w+\Delta g^{s\top}_k {T}^\top_g {T}_g \Delta v
+\Delta g^{s\top}_k {T}^\top_g {T}_f \Delta w+\Delta v^\top {T}^\top_g {T}_f \Delta w)=\Psi^{\top}_k {\mathcal{T}}\Psi_k$,
 where the inequality holds because $\Delta f^{s\top}_k\Delta f^s_k=\|\Delta f^s_k\|^2 \leq L^2_{f_d} \|\Delta^s_k\|^2$ (and similarly for $\Delta g^{s\top}_k\Delta g^s_k$) by Lemma \ref{lem:lip-dec} and $\Delta g^{s\top}_k{T}^\top_g {T}_g \Delta g^s_k\leq \lambda_{\max}({T}^\top_g {T}_g)\Delta g^{s\top}_k\Delta g^s_k=\lambda_{\max}({T}^\top_g {T}_g)\|\Delta g^s_k\|^2 \leq L^2_{g_d}\lambda_{\max}({T}^\top_g {T}_g) \|\Delta^s_k\|^2$ by using the \emph{Rayleigh Quotient} and Lemma \ref{lem:lip-dec}. Now, by the Lyapunov Theorem, stability is satisfied if ${\mathcal{T}} \preceq 0$.\\[-0.35cm]  

\noindent\textbf{
Condition (iii):}
Similarly, we consider a {candidate} 
Lyapunov function $V_k=\Delta^{s\top}_k P \Delta^{s}_k$, where $P \succ 0$, which can be shown to satisfy $\Delta V_k \triangleq V_{k+1}-V_k \leq 0$, 
as follows. 
Let $\hat{\Delta} \eta_k \triangleq \begin{bmatrix}{\Delta} \tilde{f}^{s\top}_k & \Delta^{s\top}_k & {\Delta} z^\top \end{bmatrix}^\top$ and 
note that  
${\Delta} \tilde{f}^{s\top}_k\Lambda {\Delta}  \tilde{f}^s_k \le {\Delta}  \tilde{f}^{s\top}_k {\Delta}  \tilde{f}^s_k \leq \mathcal{L}^2{\Delta}^{s\top}_k {\Delta}^s_k$, where the inequalities hold by {choosing $\Gamma$ such that} $\Gamma \triangleq I -\Lambda \succeq 0$ and Lemma \ref{lem:lip-dec}, respectively. Hence, $\mathcal{L}^2\Delta^{s\top}_k \Delta^s_k-{\Delta} \tilde{f}^{s\top}_k\Lambda {\Delta}  \tilde{f}^s_k \ge 0$. Then, inspired by a trick used in \cite[Proof of Theorem 1]{delshad2016robust}, to satisfy $\Delta V_k \le 0$, it suffices to guarantee that $\tilde{V}_k \triangleq \Delta V_k +\mathcal{L}^2\Delta^{s\top}_k \Delta^s_k-{\Delta}  \tilde{f}^{s\top}_k\Lambda {\Delta} \tilde{f}^s_k = \Delta V_k +\mathcal{L}^2\Delta^{s\top}_k \Delta^s_k-{\Delta}  \tilde{f}^{s\top}_k(I-{\Gamma}) {\Delta}  \tilde{f}^s_k \le 0$, where 
$\tilde{V}_k ={\Delta}  \tilde{f}^{s\top}_k P {\Delta}  \tilde{f}^{s}_k \hspace{-0.1cm}+\hspace{-0.1cm}{\Delta} z^\top P {\Delta} z+\hspace{-0.1cm}2{\Delta} z^\top P{\Delta}  \tilde{f}^s_k\hspace{-0.1cm}-\hspace{-0.1cm}\Delta^{s\top}_k P \Delta^s_k
+\mathcal{L}^2\Delta^{s\top}_k \Delta^s_k-{\Delta}  \tilde{f}^{s\top}_k (I-{\Gamma}) {\Delta}  \tilde{f}^s_k
= {\Delta}  \tilde{f}^{s\top}_k(P+\Gamma-I){\Delta} \tilde{f}^s_k+\Delta^{s\top}_k(\mathcal{L}^2I-P) \Delta^s_k
+{\Delta} {z}^\top P {\Delta} z+2{\Delta} {z}^\top P \Delta  \tilde{f}^s_k ={\Delta} \eta^\top_k \mathcal{P} {\Delta} \eta_k \le 0$.
\end{proof}
\moh{Finally, we will provide upper bounds for the interval widths and compute their \yong{steady-state values, if they exist}.}
\begin{lem}[Upper Bounds of the Interval Widths and their Convergence]\label{lem:convergence}
Consider the system \eqref{eq:system} and the GSISIO observer \eqref{eq:xup}--\eqref{eq:dup}. Suppose all assumptions in Theorem \ref{thm:existence} \moha{hold}. 
Then, \mohamm{there exist uniformly bounded upper sequences $\{\delta^x_{k},\delta^d_{k-1}\}_{k=1}^{\infty}$ for} interval width sequences $\{\|\Delta^x_{k}\|,\|\Delta^d_{k-1}\|\}_{k=1}^{\infty}$, \mohamm{which can be computed} 
\moha{as follows}:
\begin{align*}
 &\|\Delta^x_{k}\|\hspace{-.1cm} \leq \hspace{-.1cm}\delta^x_{k} \hspace{-.1cm}= \hspace{-.1cm} \mathcal{{L}}^k \delta_0^x \hspace{-.1cm}+\hspace{-.1cm} \|{\Delta} z\| \hspace{-.1cm}\left(\hspace{-.1cm}\frac{1-{\mathcal{L}}^k}{1-{\mathcal{L}}} \hspace{-.1cm}\right)\hspace{-.1cm}, \|\Delta^d_{k-1}\|\hspace{-.1cm} \leq \hspace{-.1cm}\delta^d_{k-1} \hspace{-.1cm}=\hspace{-.1cm} \mathcal{G}(\delta^x(k)\hspace{-.05cm})\hspace{-0.05cm}
\end{align*}
where 
$\mathcal{G}(x)\triangleq ((1+L_{f_d}) \|\hat{J}_1\|+L_{g_d}\|\hat{J}_2\|)x+ \|\hat{J}_1 \Delta w+\hat{J}_2 \Delta v\|$, $\Delta z = {T}_f\Delta w+{T}_g\Delta v$, $\Delta w \triangleq \overline{w}-\underline{w}$, $\Delta v \triangleq \overline{v}-\underline{v}$, $\hat{J} \triangleq \begin{bmatrix} \hat{J}_1&  \hat{J}_2 \end{bmatrix}\triangleq J^++J^{++}$ and $L_{f_d}$, $L_{g_d}$, ${T}_f ,{T}_g$ are given in Lemma \ref{lem:lip-dec} and Theorem \ref{thm:boundedness}. Furthermore, if {Condition} \eqref{item:first} in Theorem \ref{thm:boundedness} holds {with strict inequality}, then the upper bound sequences 
converge to steady-state values as \yongz{follows}: 
\begin{align*}
\overline{\delta}^x &\triangleq \lim_{k \to \infty} \delta^x_{k} = \|\Delta z\|\frac{\mathcal{L}}{1-\mathcal{L}}, \quad \overline{\delta}^d \triangleq \lim_{k \to \infty} \delta^d_{k} =\mathcal{G}(\overline{\delta}^x).
\end{align*}
\moh{\yong{On the other hand}, if \yongn{Condition} \eqref{item:third} or \eqref{item:second} in Theorem \ref{thm:boundedness} hold, then the interval widths $\|\Delta^x_{k}\|$ and $\|\Delta^d_{k}\|$ are uniformly bounded by \yong{$\min\{\|\Delta^x_{0}\|,\Delta^P_0\}$ and $\min\{\mathcal{G}(\|\Delta^x_{0}\|),\mathcal{G}((\Delta^P_0)\}$,} 
respectively, \yong{with $\Delta^P_0 \triangleq \min \limits_{P \in \mathbb{P}} \sqrt{\frac{(\Delta^x_{0})^\top P \Delta^x_{0}}{\lambda_{\min}(P)}}$, where $\mathbb{P}$ is the set of all $P$ that} solve the LMI in \yongn{Condition} \eqref{item:second}.} 
\end{lem}
\begin{proof}
Applying $\|\Delta^x_{k}\| \leq \mathcal{L} \|\Delta^x_{k-1}\|+\|\Delta z\|$ repeatedly, 
yields
\begin{align*}
\|\Delta^x_{k}\|\hspace{-.05cm}\leq\hspace{-.05cm} {\mathcal{L}}^k\|\Delta^x_0\|\hspace{-.05cm}+\hspace{-.05cm}\textstyle{\sum}_{i=0}^{k-1}{\mathcal{L}}^{k-i}\|{\Delta} z\|\hspace{-.05cm}= \hspace{-.05cm}{\mathcal{L}}^k \delta_0^x \hspace{-.05cm}+\hspace{-.05cm} \|{\Delta} z\| \frac{1-\mathcal{L}^k}{1-{\mathcal{L}}}. 
\end{align*}
Further, from \eqref{eq:dup}--\eqref{eq:hlow}: 
$\Delta^d_{k-1}
\leq \hat{J}_1(\Delta^x_k+\Delta f^x_{k})+\hat{J}_2\Delta g^x_{k}+\hat{J}_1\Delta w+\hat{J}_2\Delta v$, where $\hat{J} \triangleq \begin{bmatrix} \hat{J}_1&  \hat{J}_2 \end{bmatrix} \triangleq |J|$. The rest of the proof is similar to the one for \cite[Lemma 2]{khajenejad2020simultaneous}.
\end{proof}
\vspace{-0.3cm}
\section{Illustrative Example} \label{sec:examples}\vspace{-0.05cm}
We consider a slightly modified version of a nonlinear system in \cite{de2019robust}, \yongs{without the uncertain matrices, with the inclusion of} unknown inputs, and with the following parameters (cf. \eqref{eq:system}):
$n=l=p=2$, $m=1$, $f(x_k)=\begin{bmatrix} f_1(x_k) & f_2(x_k)\end{bmatrix}^\top$,  $g(x_k)=\begin{bmatrix} g_1(x_k) & g_2(x_k)\end{bmatrix}^\top$, $B=D=0_{2 \times 1}$,  $G=\begin{bmatrix} 0 & -0.1 \\ 0.2 & -0.2 \end{bmatrix}$, $H=\begin{bmatrix} -0.1 & 0.3 \\ 0.25 & -0.75 \end{bmatrix}$, $\overline{v}=-\underline{v}=\overline{w}=-\underline{w}=\begin{bmatrix} 0.2 & 0.2 \end{bmatrix}^\top$, $\overline{x}_0=\begin{bmatrix} 2 & 1.1 \end{bmatrix}^\top$, $\underline{x}_0=\begin{bmatrix} -1.1 & -2 \end{bmatrix}^\top$ \yongs{with 
\begin{align*}
\begin{array}{rl}
f_1(x_k)&=0.6x_{1,k} - 0.12x_{2,k}+1.1 \sin(0.3x_{2,k}-.2x_{1,k}),\\ f_2(x_k)&=-0.2x_{1,k}-0.14x_{2,k}, \ g_2(x_k)=\sin(x_{1,k}),\\ g_1(x_k)&=0.2x_{1,k}+0.65x_{2,k}+0.8\sin(0.3x_{1,k}+0.2x_{2,k}), 
\end{array}
\end{align*}
while} the unknown input signals are depicted in Figure \ref{fig:variances2}. 

Note that rk$(H)=1 \yongs{\,< 2=p}$, thus the feedthrough matrix is not full rank and hence, the approach in \cite{khajenejad2020simultaneous} is not applicable. Moreover, applying \cite[Theorem 1]{singh2018mesh}, 
upper and lower bounds for partial derivatives of $f(\cdot)$ and $g(\cdot)$ \mohamm{are obtained as}: $(a^f_{11},a^f_{12},a^f_{21},a^f_{22})=(0.38, -0.52,-0.2-\epsilon,-0.14-\epsilon)$, $(b^f_{11},b^f_{12},b^f_{21},b^f_{22})=(0.82,0.21, -0.2+\epsilon, -0.14+\epsilon)$, $(a^g_{11}, a^g_{12}, a^g_{21}, a^g_{22})=( -0.04, 0.49, -1, -\epsilon)$ and $(b^g_{11},b^g_{12},b^g_{21},b^g_{22})=(0.44,0.81,1, \epsilon)$, where $\epsilon$ is a very small positive value, ensuring that the partial derivatives are in open intervals (cf. \cite[Theorem 1]{yang2019sufficient}). Moreover, $L_f=0.35$ and $L_g=0.74$ and Assumption \ref{assumption:mix-lip} holds by \cite[Theorem 1]{yang2019sufficient}). Furthermore, computing $K=\begin{bmatrix} K_1 & K_2 \end{bmatrix}=\begin{bmatrix} 0.0267 & 0 & 0.0666 & 0.1061 \\ 0.4177 & 2.1203 & 1.0817 & 2.0209 \end{bmatrix}$ and $L=\begin{bmatrix} L_1 & L_2 \end{bmatrix}=\begin{bmatrix} 0 & 0.1017 & 0 & 0 \\ 0.5194 & 1.1814 & 1.2787 & 1.9302\end{bmatrix}$, we obtain ${\rm rk}(I-K_1-L_1)={\rm rk}(I-K_1+L_1)=2$.
Therefore, by Theorem \ref{thm:existence}, the existence of correct framers is guaranteed, i.e., the true states and unknown inputs are within the estimate intervals. This, can be verified from Figure \ref{fig:variances2} that depicts interval estimates as well as the true states and unknown inputs. 
\begin{figure}[t]
\begin{center}
\includegraphics[scale=0.150,trim=48mm 0mm 5mm 10mm,clip]{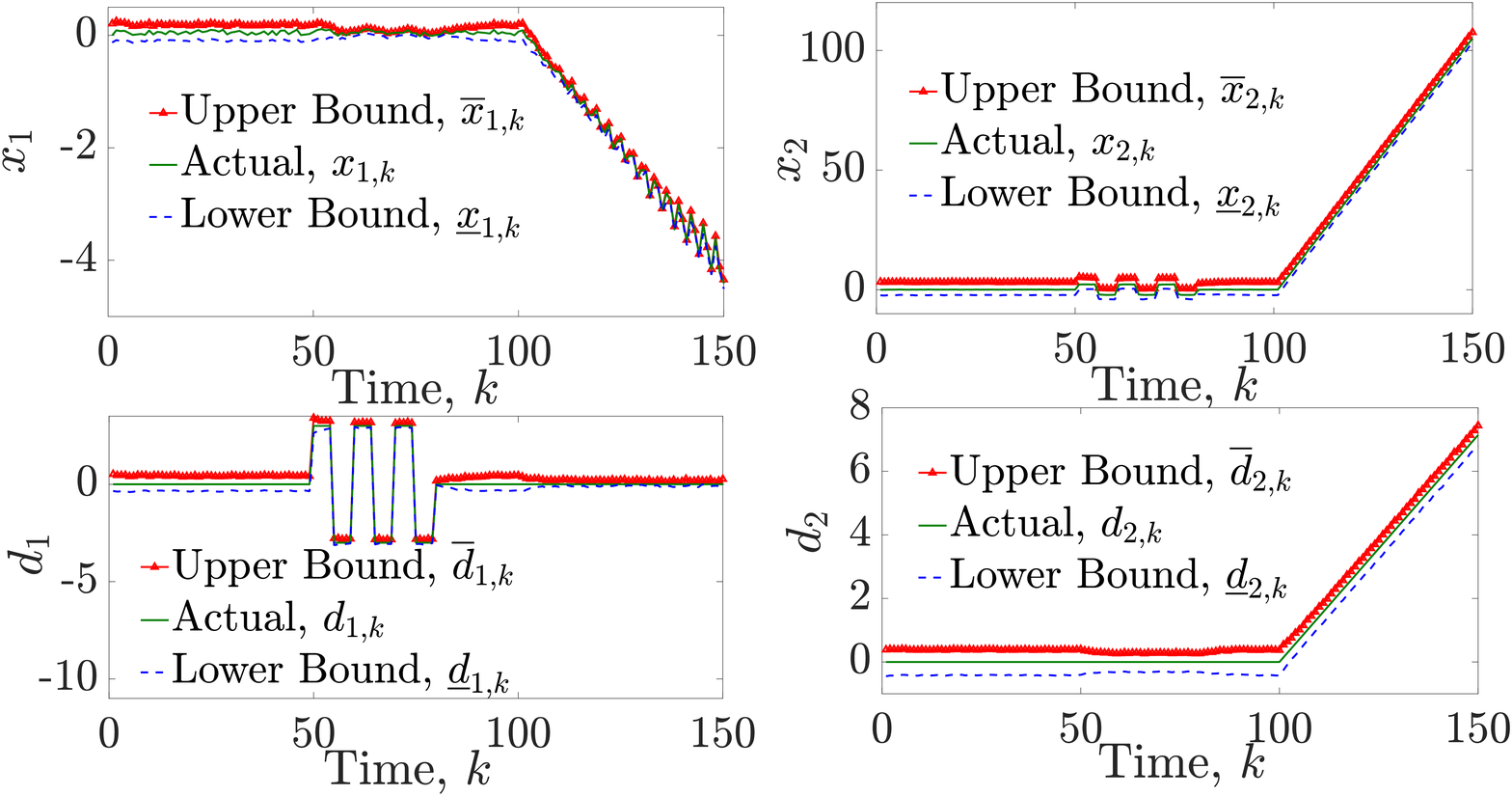}
\vspace{-0.4cm}
\caption{Actual states and inputs, $x_{1,k}$, $x_{2,k}$, $d_{1,k}$, $d_{2,k}$, as well as their estimated maximal and minimal values, $\overline{x}_{1,k}$, $\underline{x}_{1,k}$, $\overline{x}_{2,k}$, $\underline{x}_{1,k}$, $\overline{d}_{1,k}$, $\underline{d}_{1,k}$, $\overline{d}_{2,k}$, $\underline{d}_{2,k}$. \label{fig:variances2}}
\end{center}
\vspace{-0.1cm}
\end{figure} 
\begin{figure}[t]
\begin{center}
\includegraphics[scale=0.147,trim=38mm 10mm 10mm 20mm,clip]{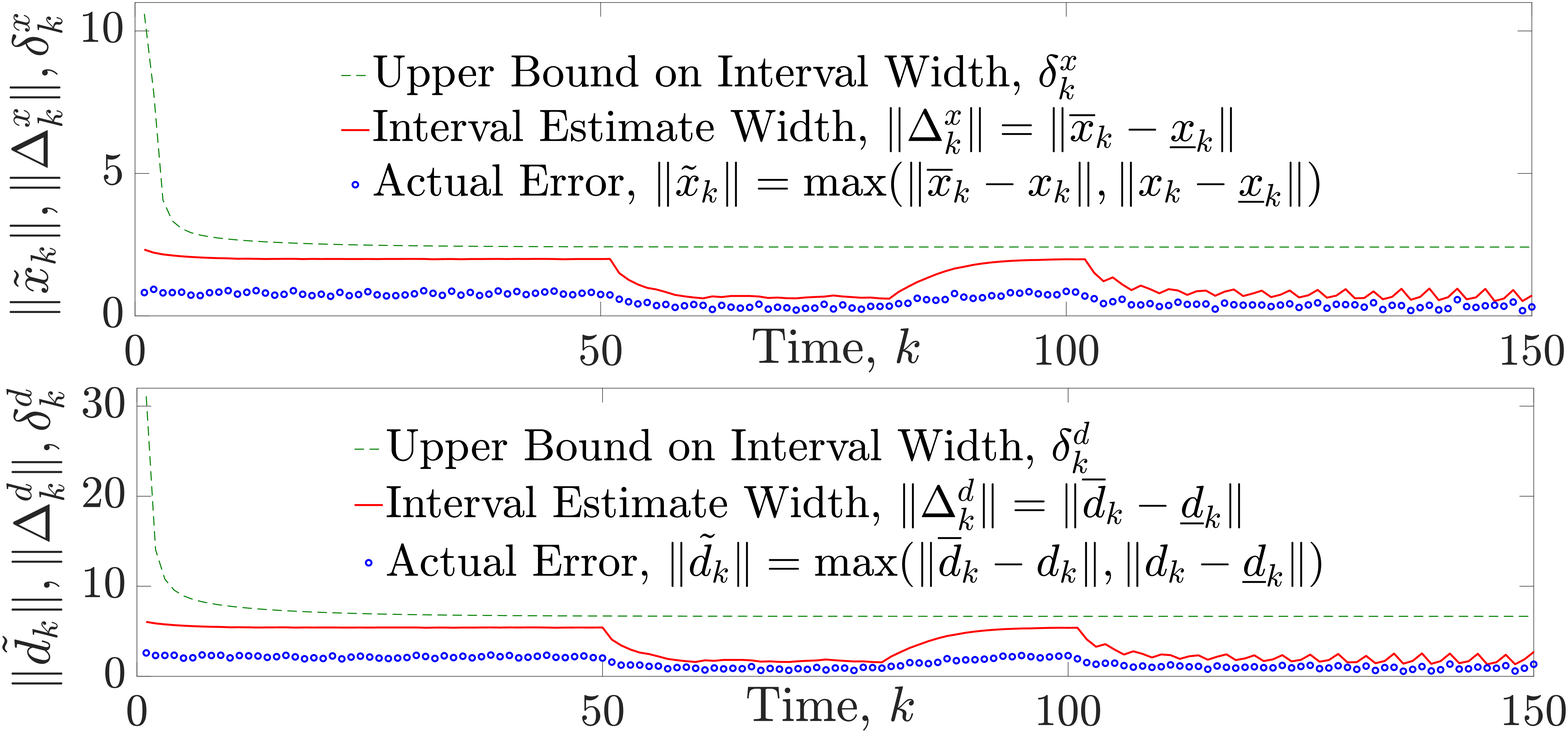}
\vspace{-0.5cm}
\caption{Estimation errors, estimate interval widths and their upper bounds for the interval-valued estimates of states, $\|\tilde{x}_{k|k}\|$, \hspace{-0.05cm}$\|\Delta ^x_k\|$, \hspace{-0.05cm}$\delta^x_k$, and unknown inputs, $\|\tilde{d}_{k}\|$, $\|\Delta ^d_k\|$, $\delta^d_k$. \label{fig:variances3}}
\end{center}
\vspace{-0.15cm}
\end{figure}
In addition, from \cite[(10)--(13)]{yang2019sufficient}), we obtain $C_f=\begin{bmatrix}  0.251 & 0 \\ 0.0029 & 0.201\end{bmatrix}$, $C_g=\begin{bmatrix} 0 & 0.225 \\ -.374 & -.045 \end{bmatrix}$, which implies that ${L}_{f_d}=0.852$ and $L_{g_d}=1.19$ by Lemma \ref{lem:lip-dec}. Consequently, ${\mathcal{L}}=0.643$ 
satisfies 
Condition \eqref{item:first} in Theorem \ref{thm:boundedness}. 
So, we expect to obtain uniformly bounded estimate errors with convergent upper bounds. Figure \ref{fig:variances3} \mohamm{illustrates this}, where at each step, the actual error is less than or equal to the interval width, which in turn is less than or equal to the predicted upper bound for the interval width and the upper bounds converge to some steady-state values. \yong{Note that, despite our best efforts, we were unable to find interval-valued observers \mohamm{for nonlinear systems} in the literature that simultaneously return state and unknown input estimates for comparison with our results.} 
\balance
\section{Conclusion} \label{sec:conclusion}
In this paper, a simultaneous  input and state interval-valued observer was proposed for bounded-error mixed monotone Lipschitz nonlinear systems with unknown inputs \yongz{and rank-deficient feedthrough}. We derived \mohamm{necessary and} sufficient conditions for the existence \mohamm{and correctness} of our observer \mohamm{as well as the tightness of} 
the input interval estimates. 
\moham{Further}, several conditions for the stability of the observer, i.e., the uniform boundedness of the interval widths were 
\mohamm{provided and} the effectiveness of the proposed approach \mohamm{was demonstrated} with an example. 
\mohamm{\yongz{Future work will seek} 
tighter decomposition functions 
and necessary conditions for the interval estimator 
stability. }
\vspace{-0.05cm}

{
\bibliographystyle{unsrturl}

\bibliography{biblio}
}
\end{document}